\documentclass[11pt,a4paper]{article}

\usepackage[T1]{fontenc}
\usepackage{latexsym}
\usepackage{fullpage}
\usepackage{authblk}
\usepackage{amsthm}
\usepackage{amssymb}
\usepackage[utf8]{inputenc}
\usepackage{amsmath}
\usepackage{graphicx}
\usepackage{mathtools}
\usepackage{enumerate}

\newtheorem{theorem}{Theorem}
\newtheorem{lemma}[theorem]{Lemma}
\theoremstyle{plain}
\newtheorem{observation}[theorem]{Observation}

\newcommand{\ignore}[1]{}

\newcommand{\tree}[2]{\ensuremath{#1[#2]}}
\newcommand{\hdeg}{\ensuremath{\mathsf{heavydeg}}}
\newcommand{\copt}[1]{\mathrm{COPT}(#1)}
\newcommand{\pdfs}[1]{\mathrm{PDFS}(#1)}
\usepackage{enumitem}
\newcommand{\cS}{\mathcal{S}}
\newcommand{\cC}{\mathcal{C}}
\newcommand{\card}[1]{\left|#1\right|}
\newcommand{\problemTPE}{\textup{ECTE}}
\newcommand{\routelen}[1]{\ell(#1)}
\newcommand{\cost}[1]{\xi(#1)} 
\newcommand{\algAdversarial}[2][]{\mathrm{ADFS}_{#1}(#2)} 
\newcommand{\noroutes}[1]{\card{#1}}
\newcommand{\allB}{\mathcal{B}}
\newcommand{\Tprime}[1]{T_{#1}}
\newcommand{\classEps}{\mathcal{I}}
\newcommand{\myst}{\hspace{0.1cm}\bigl|\bigr.\hspace{0.1cm}}
\newcommand{\minRoutesStr}{\mathcal{R}}
\newcommand{\DFSroute}{R_{\textup{DFS}}}

\title{Energy Constrained Depth First Search\footnote{Research partially supported by National Science Centre (Poland) grant number 2015/17/B/ST6/01887.}}

\author[1]{Shantanu Das}
\author[2]{Dariusz Dereniowski}
\author[3]{Przemys\l{}aw~Uzna\'nski}
\affil[1]{LIF, Aix-Marseille University and CNRS, Marseille, France\\
  \texttt{shantanu.das@lif.univ-mrs.fr}}
\affil[2]{Faculty of Electronics, Telecommunications and Informatics, \mbox{Gda{\'n}sk University of~Technology}, Poland\\
  \texttt{deren@eti.pg.edu.pl}}
\affil[3]{Department of Computer Science,
ETH Z\"urich, Switzerland\\
  \texttt{przemyslaw.uznanski@inf.ethz.ch}}

\date{}

\begin{document}

\maketitle

\begin{abstract}
Depth first search is a natural algorithmic technique for constructing a closed route that visits all vertices of a graph. The length of such route equals, in an edge-weighted tree, twice the total weight of all edges of the tree and this is asymptotically optimal over all exploration strategies. This paper considers a variant of such search strategies where the length of each route is bounded by a positive integer $B$ (e.g. due to limited energy resources of the searcher). The objective is to cover all the edges of a tree $T$ using the minimum number of routes, each starting and ending at the root and each being of length at most $B$. 
To this end, we analyze the following natural greedy tree traversal process that is based on decomposing a depth first search traversal into a sequence of limited length routes. 
Given any arbitrary depth first search traversal $R$ of the tree $T$, we cover $R$ with routes $R_1,\ldots,R_l$, each of length at most $B$ such that: $R_i$ starts at the root, reaches directly the farthest point of $R$ visited by $R_{i-1}$, then $R_i$ continues along the path $R$ as far as possible, and finally $R_i$ returns to the root. We call the above algorithm \emph{piecemeal-DFS} and we prove that it achieves the asymptotically minimal number of routes $l$, regardless of the choice of $R$. Our analysis also shows that the total length of the traversal (and thus the traversal time) of piecemeal-DFS is asymptotically minimum over all energy-constrained exploration strategies. 
The fact that $R$ can be chosen arbitrarily means that the exploration strategy can be constructed in an online fashion when the input tree $T$ is not known in advance. Each route $R_i$ can be constructed without any knowledge of the yet unvisited part of $T$. Surprisingly, our results show that depth first search is efficient for energy constrained exploration of trees, even though it is known that the same does not hold for energy constrained exploration of arbitrary graphs.
\end{abstract}


\section{Introduction}
Graph-theoretic problems in which one wants to cover the entire graph with one or more routes satisfying certain objective is a well established and long studied topic in many areas of computer science.
Particular problems vary depending on the research area or potential applications, including the study of simple graph traversals like DFS or BFS, for algorithmic purposes, to  complex transportation problems with many variations of traveling salesman problem (TSP), or pursuit-evasion games like the watchman problem, and finally distributed monitoring of networks using mobile agents.

For one possible application of our results, consider a mobile robot that needs to explore an initially unknown tree.
We assume that the tree is edge-weighted and the weight of each edge denotes the length of that edge.
Starting from a single vertex (the root) of the tree, the robot must traverse all edges of and return to its initial location.
Upon visiting a vertex $v$ for the first time, the robot discovers the edges incident to $v$ and can choose one of them to continue the exploration.
Provided that the robot can remember the visited vertices and edges, a simple depth first search (DFS) is an efficient algorithm for exploring the tree, achieving the optimal cost of twice the sum of the lengths of edges in the tree.

Consider a more interesting scenario, when the robot has a limited source of energy (e.g. a battery) which allows it to traverse a path of length at most $B$ (we say such a robot is \emph{energy constrained}).
Naturally, we assume that each vertex of the tree is at distance at most $B/2$ from the root, otherwise the tree cannot be fully explored.
In this case, exploration is possible if the robot can recharge its battery whenever it returns back to the starting location.
Thus, the exploration is a collection of routes of the robot, each of which starts and ends at the root, and has length at most $B$. We are interested in the minimum number of such routes needed (i.e. the number of times the robot has to recharge) to completely explore the tree.

This model of exploration may be of interest for several reasons. One obvious reason is related to the capabilities of the robot; it may have a restricted fuel tank capacity or perhaps a harsh or risky environment enforces a return to its home-base every so often.
A robot that returns periodically to the root can inform about new discoveries --- in this way the knowledge is accumulated gradually at the base-station while the algorithm progresses. This may, for example, reduce the risk of having no data in case of robot failure prior to the end of exploration. From a different point of view, this process may be seen as a \emph{piecemeal learning}, that is, one in which it is possible to have a trade off between exploration and utilization; the two phases representing parts in which learning occurs (exploration) and part in which accumulated knowledge is used (utilization).
Finally, having many restricted-length routes covering a tree instead of a single long route may be potentially applied in scenarios in which one wants to minimize exploration time by strategies using multiple robots. In fact, when the robots are incapable of refueling, we can use several robots to explore the tree, each robot traversing a path of length $B$. In that case, it is important to minimize the number of robots used as well as the total energy cost for exploration.

Note that the \emph{piecemeal exploration} problem has been studied before not just for trees but also for arbitrary connected graphs. However those results were restricted to visiting vertices at depth of at most $B/2(1+\beta)$, for some $\beta > 0$, with the cost of exploration deteriorating sharply as $\beta$ approached zero. In this paper we would like to consider exploration strategies that completely visit all trees up to the maximum possible depth of $B/2$. No such exploration algorithm have been studied for either general graphs or special graphs such as trees. Simple strategies based on depth-first search (breadth first search) perform badly in the case of piecemeal exploration of arbitrary graphs. However as we show in this paper, the piecemeal version of depth-first search performs optimally in trees. This fact is surprising given the fact that for exploration by multiple open routes (routes that do not end at the root) depth-first strategies in trees can have an overhead of $\Omega(\log{n})$~\cite{DDK15}.     

\medskip

\noindent \textbf{Related work:}
There exists extensive literature on graph traversal and exploration, we survey here only the most relevant results on graph exploration by mobile agents.
Exploration of general graphs having $n$ nodes and $m$ edges, by a single agent, has been studied in \cite{PanaiteP99} who gave an algorithm of $m+O(n)$ steps. For exploration by $k$ agents, \cite{FraigniaudGKP06} provides an exploration algorithm taking $O(D+n/\log k)$ steps in trees of $n$ nodes and height $D$. This algorithm turns out to be $O(k/\log k)$ competitive~\cite{HigashikawaKLT14} (where competitiveness is the ratio of the number of steps of an algorithm over the optimal number of steps).
Authors in~\cite{BrassCGX11} give a $O(n/k+D^{k-1})$ time algorithm for tree exploration while \cite{DyniaKHS06} gives an algorithm for sparse trees with competitive ratio $O(D^{1-1/p})$, where $p$ is defined as the tree density. For some lower bound on exploration time, see \cite{DyniaLS07,FraigniaudGKP06,HigashikawaKLT14}.
For other recent results on exploration time see e.g.~\cite{BrassCGX11,DisserMNSS16,DobrevKM12,MegowMS12,OrtolfS14}.
Other than optimizing time, exploration using little memory for the agents has also been studied, see e.g.~\cite{AmbuhlGPRZ11,DisserHK16}.

None of the results mentioned above consider any energy limitation for the agents.
The energy constrained exploration problem was first studied under the name of \emph{Piecemeal Graph Exploration} ~\cite{BetkeRS95}, with the assumption that the route length $B \geq 2(1+\beta)r$, where $r$ is the furthest distance from the starting node to other nodes, and $0<\beta<1$. That paper provided exploration algorithms for a special class of grid graphs with `rectangular obstacles'. 
Awerbuch et al.~\cite{AwerbuchBRS99} showed that, for general graphs, there exists an energy constrained exploration algorithm with a total cost of $O(m+n^{1+o(1)})$. This has been further improved (by an algorithm that is a combination of DFS and BFS) to $O(m+n\log^2 n)$ in~\cite{AwerbuchK98}.
Finally \cite{DuncanKK06} provided an exploration algorithm for general unknown weighted graphs with total cost asymptotic to the sum of edge weights of the graph.
Note that, as mentioned, all the above strategies require the length of each route to be strictly larger than the shortest return path from the starting vertex to the farthest vertex. In other words, these algorithms fail in the extreme cases when the height of the explored tree (or the diameter of the graph) is equal to half of the energy budget, which seem to be the most challenging cases.

The same tree exploration model as we study in this work has been considered in \cite{DDK15,DyniaKS06} for unweighted trees and multiple agents, with one difference: each agent traverses a path of length at most $B$ such that the path starts at the root but may end at any node of the tree (in other words, agents do not have to return to the homebase).
It has been shown in~\cite{DDK15} that if the tree is not known in advance, then there exists an exploration algorithm (that minimizes the number of agents used) with competitive ratio of $O(\log B)$ and this is the best possible. On the other hand, it was shown that by allowing the route lengths to be a constant factor more than $B$, it is possible to explore the tree using the minimum number of agents~\cite{DyniaKS06}.
Distributed algorithms for energy constrained agents has been a subject of recent investigation, see e.g.~\cite{AnayaCCLPV12,AnayaCCLPV16,BartschiC0DGGLM16}.
Authors in~\cite{CzyzowiczDMR16} also consider a model in which agents are allowed to transfer part of their energy to another agent.
There are many studies on variants of the traveling salesman problem, including the $k$-TSP \cite{Arora98,FredericksonHK78} related to the task of finding a bounded length route in a graph. Such results are out of scope for this paper.

\medskip

\noindent \textbf{Our Results and Outline:} 
In this work we analyze a very natural process of partitioning a depth first search traversal $\DFSroute$ of a tree into a sequence $\cS=(R_1,\ldots,R_k)$ of routes where each route $R_i$ has length at most $B$, starts and ends at the root of the tree (see Section~\ref{sec:dfs} for a formal definition).
We prove that the number of routes $k$ is asymptotically optimal (Theorem~\ref{thm:num-agents}), that is, it is within a constant factor of the number of routes in any exploration strategy composed of routes of length at most $B$ that cover the entire tree.
This fact, being intuitively expected for trees (although it does not hold in general graphs~\cite{DuncanKK06}) turns out to be nontrivial.
Our approach is to consider another parameter of an exploration strategy: the \emph{cost} (see Section~\ref{sec:results}) defined as the sum of the lengths of all routes in an exploration strategy.
In order to prove our main result, we argue, in Section~\ref{sec:analysis}, that the cost of $\cS$ is asymptotically optimal (see Theorem~\ref{thm:cost}).
Then, in Section~\ref{sec:results} we argue that the fact that $\cS$ has small cost implies that the number of routes in $\cS$ is expectedly small.

We emphasize that the above claim holds independently of the choice of the initial depth first search traversal $\DFSroute$.
The implications of this fact are twofold.
First, it provides a theoretical insight into such a partitioning of a depth first search traversals into bounded-length segments.
Second, for an exploration algorithm design it means that the routes $R_i$ may be constructed without knowing $\DFSroute$ in advance, or more precisely, the routes may be build in an online fashion based only on the knowledge of the subtree explored to date.
This property makes our algorithm suitable for online exploration of unknown trees by energy constrained mobile agents.

\section{Exploration strategies} \label{sec:preliminaries}

In this work we consider edge-weighted rooted trees $T=(V(T),E(T),\omega\colon E(T)\to\mathbb{R}_+)$, with root $r$.
We define a \emph{route} $R$ as sequence of nodes,
$R=(v_0,v_1,\ldots,v_l)$, where $v_i$ is a vertex of $T$ for each $i\in\{0,\ldots,l\}$, as follows:
\begin{enumerate} [label={\normalfont{(\roman*)}},leftmargin=2\parindent]
 \item $\{v_i,v_{i+1}\}\in E(T)$ for each $i\in\{0,\ldots,l-1\}$,
 \item $v_0=v_l$ is the root $r$ of $T$.
\end{enumerate}

Informally speaking, a route is a sequence of vertices forming a walk in $T$ that starts and ends at the root. We define the \emph{length} of $R$ to be
\[\routelen{R}=\sum_{i=1}^{l}\omega(\{v_{i-1},v_i\}).\]
We say that a vertex $v$ is \emph{visited} (and edge $\{v_{i-1},v_i\}$ is traversed) 
by the route if $v=v_i$ for some $i\in\{1,\ldots,l\}$.
We also say that the subtree of $T$ composed with all vertices visited by $R$ is \emph{covered} by the route.

Given a tree $T$ and an integer $B$, we say that $\cS=(R_1,\ldots,R_k)$ is a $B$-\emph{exploration strategy} for $T$ (or simply \emph{exploration strategy} if $B$ is clear from the context) if for each $i\in\{1 ,\ldots, k\}$, $R_i$ is a route in $T$ of length at most $B$, and each vertex of $T$ is visited by some route in $\cS$.
We write $\noroutes{\cS}$ to refer to the number of routes in $\cS$, $k=\noroutes{\cS}$.

\section{Problem statement and DFS exploration} \label{sec:dfs}

The formulation of the combinatorial problem, to which we refer as \emph{energy constrained tree exploration}, we study in this work is as follows.

\begin{description}
\item[Energy Constrained Tree Exploration problem] ($\problemTPE$)\\[2mm]  
Given a real number $B>1$ and an edge-weighted rooted tree $T$ of height at most $B/2$ what is the minimum integer $k$ such that there exists a $B$-exploration strategy that consists of $k$ routes?
\end{description}

\noindent Our goal is to analyze a particular type of solution to this problem, namely, an exploration strategy that behaves like a depth first search traversal but adopted to the fact that route lengths are bounded by $B$.
Let $\DFSroute=(v_0,v_1,\ldots,v_l)$ be a route in $T$ that covers the tree $T$ and performs a depth first search traversal of $T$.
(Note that $\DFSroute$ is a route and thus we consider a depth first search traversal to have node repetitions.)
For two vertices $u$ and $v$ of $T$, $d(u,v)$ denotes the distance between $u$ and $v$ understood as the sum of weights of the edges of the path connecting these vertices.
We refer by $\pdfs{T}=(R_1,\ldots,R_k)$ (\emph{Piecemeal Depth First Search}) to the following $B$-exploration strategy constructed iteratively for $i:=1,\ldots,k$ (see also Figure~\ref{fig:strategy} for an example):
\begin{enumerate} [label={\normalfont{(\roman*)}},leftmargin=*]
\item let $j_0=0$ i.e. $v_{j_0} = v_0 = r$,
\item\label{it:dfs2} $R_i$ continues DFS exploration from where $R_{i-1}$ \emph{stopped making progress} (from the node $v_{j_{i-1}}$) as long as for currently visited $v_p$: 
\begin{equation} \label{eq:route-len}
d(r,v_{j_{i-1}}) + \routelen{(v_{j_{i-1}},v_{j_{i-1}+1}, \ldots, v_p)} + d(v_p,r) \le B,
\end{equation}
\item furthest $v_p$ (for $p \le l$) that satisfies condition from \ref{it:dfs2} is denoted as $v_{j_i}$, the vertex where $R_i$ stopped making progress,
\item let $R_i= P_{i-1} \circ (v_{j_{i-1}},v_{j_{i-1}+1},\ldots,v_{j_i-1},v_{j_i}) \circ P_i^R$, where $P_{i-1}$ is the path from $r$ to $v_{j_{i-1}}$, and $P_i^R$ is the path from $v_{j_i}$ to $r$. 
\end{enumerate}
Such a strategy $\pdfs{T}$ is called a \emph{DFS $B$-exploration}.
We will say that the part of $R_i$ containing the subsequence $(v_{j_{i-1}},\ldots,v_{j_i})$ \emph{makes progress} on the route $\DFSroute$.
\begin{figure}[htb]
\begin{center}
\includegraphics[scale=0.6]{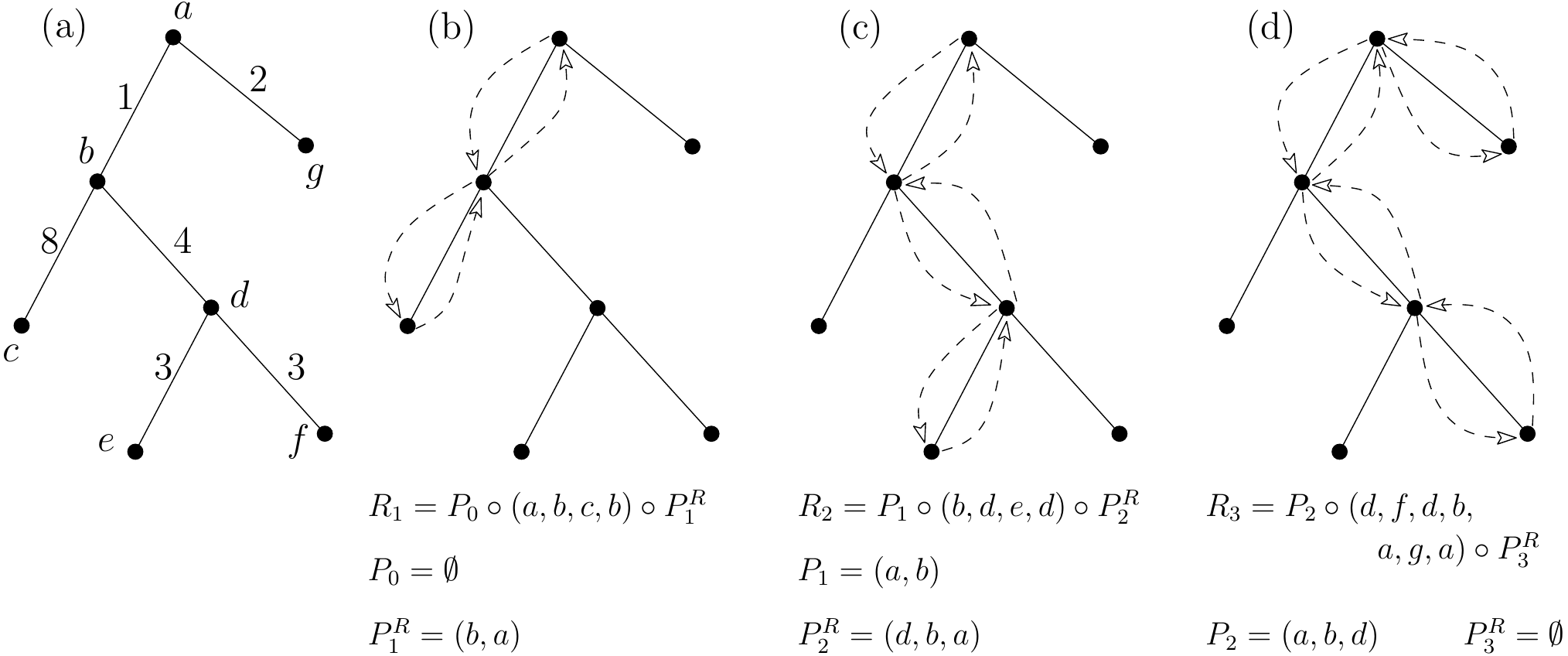}
\caption{A route $\DFSroute$ and the corresponding $\pdfs{T}=(R_1,R_2,R_3)$ with $B=20$:
         (a) a depth first search traversal $\DFSroute=(a,b,c,b,d,e,d,f,d,b,a,g,a)$;
         (b)-(d) routes $R_1,R_2,R_3$ with lengths $18,16$ and $20$, respectively}
\label{fig:strategy}
\end{center}
\end{figure}

We remark that different depth first search traversals $\DFSroute$ may result in different values of $k$ (different number of routes) in the resulting DFS $B$-exploration, although for a particular choice of $\DFSroute$ the corresponding $\pdfs{T}$ is unique.
In the rest of the work we fix the route $\DFSroute$ arbitrarily and thus $\pdfs{T}$ refers to the unique DFS $B$-exploration strategy obtained from $\DFSroute$.

\section{Our results} \label{sec:results}
The following theorem provides the first main result of this work.
\begin{theorem} \label{thm:num-agents}
Let $T$ be a tree and let the longest path from the root to a leaf in $T$ be at most $B/2$.
It holds $\noroutes{\pdfs{T}}\leq 10 \noroutes{\minRoutesStr}$, where $\minRoutesStr$ is a $B$-exploration strategy that consists of the minimum number of routes.
\end{theorem}
The theorem refers to the number of routes in an exploration strategy.
However, in order to analyze the behavior of $\pdfs{T}$, we will work with another parameter on which we will focus in the entire analysis in the subsequent section.
For any $B$-exploration strategy $\cS=(R_1,\ldots,R_k)$ of $T$ we will denote by $\cost{\cS}$ the \emph{cost} of $\cS$ defined as 
\[\cost{\cS} = \sum_{i=1}^k \routelen{R_i}.\]
We denote by $\copt{T}$ an optimal solution with respect to the cost, that is, a $B$-exploration strategy whose cost is minimum over all $B$-exploration strategies.
This strategy will serve as a reference point to prove asymptotic optimality (in terms of the number of routes) of the DFS exploration.
More precisely, we will prove the following theorem.
\begin{theorem} \label{thm:cost}
Let $T$ be a tree and let $B/2$ be greater than or equal to the longest path from the root to a leaf in $T$.
It holds $\cost{\pdfs{T}}\leq 10\cdot\cost{\copt{T}}$.
\end{theorem}
The proof is postponed to the next parts of the paper and we finish this section by concluding that Theorem~\ref{thm:cost} indeed implies Theorem~\ref{thm:num-agents}.

\medskip
\begin{proof}[Proof of Theorem~\ref{thm:num-agents}]
We start with the following observation which relates the smallest possible number of routes in a $B$-exploration strategy and the minimum possible cost.
\begin{observation} \label{obs:opt}
Given $T$ and $B$, $\noroutes{\minRoutesStr}\geq \left\lceil\cost{\copt{T}}/B\right\rceil$, where $\minRoutesStr$ is a $B$-exploration strategy with minimum number of routes.
\end{observation}
\begin{proof}
Each route of $\minRoutesStr$ is of length at most $B$.
Thus, $\noroutes{\minRoutesStr}\geq\cost{\minRoutesStr}/B$.
By definition of $\copt{T}$, $\cost{\minRoutesStr}\geq\cost{\copt{T}}$, and since $\noroutes{\minRoutesStr}$ is an integer, the claim follows.
\end{proof}
Recall that $\DFSroute=(v_0,\ldots,v_l)$ is the depth first search traversal of $T$ used to obtain $\pdfs{T}$, and the $i$-th route $R_i$ in $\pdfs{T}=(R_1,\ldots,R_k)$ makes progress on the depth first search traversal by traversing the part of $\DFSroute$ that starts at $v_{j_{i-1}}$ and ends at $v_{j_i}$.
By definition of $\pdfs{T}$, extending $R_i$ so that it makes progress with the walk $(v_{j_{i-1}},\ldots,v_{j_i},v_{j_{i+1}})$ would exceed its length to be more than $B$ for each $i<k$, i.e., $\routelen{R_i}+2\omega(\{v_{j_i},v_{j_{i+1}}\})>B$.
Consider a tree $T'$ obtained from $T$ by subdividing the edge $\{v_{j_i},v_{j_{i+1}}\}$ into two edges $\{v_{j_i},x_i\}$ and $\{x_i,v_{j_{i+1}}\}$ with weights $(B-\routelen{R_i})/2$ and $\omega(\{v_{j_i},v_{j_{i+1}}\})-(B-\routelen{R_i})/2$, respectively, for each $i\in\{1,\ldots,k-1\}$.
(Hence the sum of the two weights of the new edges $\{v_{j_i},x_i\}$ and $\{x_i,v_{j_{i+1}}\}$ equals $\omega(\{v_{j_i},v_{j_{i+1}}\})$, the weight of the subdivided edge.)
Note that the common nodes of $T$ and $T'$ (that is, the nodes of $T$) are visited by $\pdfs{T}$ and $\pdfs{T'}$ in the same order, both $\pdfs{T}$ and $\pdfs{T'}$ are $B$-exploration strategies and the length of each route in $\pdfs{T'}$, except for the last one, is of length exactly $B$.
The latter in particular implies
\begin{equation} \label{eq:no-routes-Tprime}
\noroutes{\pdfs{T'}}=\left\lceil\frac{\cost{\pdfs{T'}}}{B}\right\rceil.
\end{equation}
Note that
\begin{equation} \label{eq:cost-Tprime}
\cost{\copt{T}} = \cost{\copt{T'}}
\end{equation}
because, informally speaking, a strategy that minimizes the cost never reaches a node of degree two in order to return to previously visited node --- thus, in particular, a traversal of $\{v_{j_i},x_i\}$ is immediately followed by a traversal of $\{x_i,v_{j_{i+1}}\}$ and vice versa.

From Equation~\eqref{eq:no-routes-Tprime}, Theorem~\ref{thm:cost} applied to $T'$, Equation~\eqref{eq:cost-Tprime} and Observation~\ref{obs:opt} (used in this order) we conclude that
\[\noroutes{\pdfs{T'}}=\left\lceil\frac{\cost{\pdfs{T'}}}{B}\right\rceil \leq 10\cdot \left\lceil \frac{\cost{\copt{T'}}}{B} \right\rceil = 10\cdot \left\lceil \frac{\cost{\copt{T}}}{B} \right\rceil \leq 10\cdot\noroutes{\minRoutesStr},\]
and hence $\noroutes{\pdfs{T}}=\noroutes{\pdfs{T'}}$ completes the proof of Theorem~\ref{thm:num-agents}.
\end{proof}

\section{Bounding the cost of $\pdfs{T}$} \label{sec:analysis}

\subsection{Additional notation} \label{sec:additional-notation}

When referring to subtrees, we consider them always in the context of their distance from the root.
More specifically, we consider the \emph{potential of a node} $v$, denoted by $\varphi(v)$, to be defined as  $\varphi(v) = B/2 - d(r,v)$. 
If $u$ is the parent of $v$ in $T$, then we say that $u$ is the \emph{higher} endpoint of $\{u,v\}$ and $v$ is the \emph{lower} endpoint of the edge $\{u,v\}$; we also say that $\{u,v\}$ is a \emph{downward edge} of $u$.
For any subtree $T'$ of $T$, we define the \emph{potential} of $T'$, denoted $\varphi(T')$, to be the potential of its root.
Then, $2\varphi(T')$ is an upper bound on the total length of any route inside $T'$. 
We say that a route \emph{reaches} a potential $x$ in some subtree if it reaches a vertex having potential $x$.
Additionally, for any subtree $T'$ of $T$, we denote it \emph{weight} to be $\omega(T') = \sum_{e \in E(T')} \omega(e)$, where $E(T')$ is the edge set of $T'$.
In other words, the weight of $T$ to be the total weight of its edges.
We denote by $\tree{T}{v}$ the subtree of $T$ rooted at $v$ that contains $v$ and all its descendants, and by $\tree{T}{e}$ the tree composed of the edge $e$ and $\tree{T}{v}$ where $v$ is the lower endpoint of $e$.

We say that a subtree $T'$ of $T$ is \emph{heavy} if $\omega(T') > \varphi(T')$, and otherwise we say that $T'$ is \emph{light}.
We extend this terminology to vertices and edges: a vertex $v$ or an edge $e$ is \emph{heavy} if $\tree{T}{v}$ or $\tree{T}{e}$ is heavy, respectively.
Additionally, by $\hdeg(v)$ we denote the number of outgoing downward edges of $v$ that are heavy (note that if an edge is heavy, then both its endpoints are heavy as well).
Observe that if $v$ is any vertex of $T$ and $\tree{T}{v}$ is heavy, then one route is not enough to cover the entire $\tree{T}{v}$ in any $B$-exploration strategy.

\subsection{Adversarial DFS-exploration}
\label{sec:adversarial}

When analyzing the cost of $\pdfs{T}$ we will use a recursive approach where the $B$-exploration of any subtree $T'$ would be defined by taking $B=2\varphi(T')$, the maximum size of route starting and ending at the root of $T'$. However the first agent to reach subtree $T'$ may have performed other explorations before entering $T'$. Therefore we need to use a slightly generalized DFS $B$-exploration, called a \emph{$B'$-adversarial DFS $B$-exploration}, denoted by $\algAdversarial[B']{T}$, where the length of first route is bounded by $B'\leq B$. This is formally defined by replacing Equation~\eqref{eq:route-len} for $i=1$ in condition \ref{it:dfs2} in the definition of $\pdfs{T}$ by the following equation (see also Figure~\ref{fig:adfs}):
\begin{equation} \label{eq:route-len-new}
\routelen{(v_{j_{0}},v_{j_{0}+1}, \ldots, v_p)} + d(v_p,r) \le B'.
\end{equation}

\begin{figure}[hbt]
\begin{center}
\includegraphics[scale=0.75]{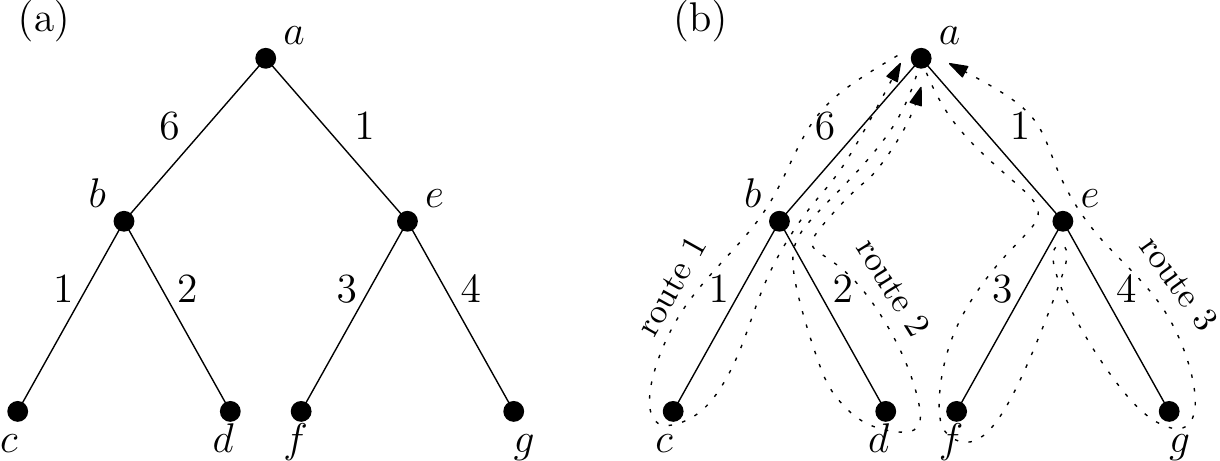}
\caption{(a) a tree with a depth first search traversal $(a,b,c,b,d,b,a,e,f,e,g,e,a)$;
         (b) a $B'$-adversarial DFS $B$-exploration $\cS$ with $B'=16$ and $B=20$ has three routes: $(a,b,c,b,a)$ (length $14\leq B'$), $(a,b,d,b,a,e,a)$ (length $18\leq B$) and $(a,e,f,e,g,e,a)$ (length $16\leq B$); $\cost{\cS}=14+16+16=46$. Note that two routes $(a,b,c,b,d,b,a)$ and $(a,e,f,e,g,e,a)$ constitute a DFS $B$-exploration strategy with cost $34$}
\label{fig:adfs}
\end{center}
\end{figure}
In other words, the length of the first route is bounded by $B'$ (Equation~\eqref{eq:route-len-new}) and the lengths of the remaining routes are bounded by $B$ (Equation~\eqref{eq:route-len} for $i>1$).

For a given tree $T$, we define an \emph{adversarial DFS exploration} of $T$, denoted by $\algAdversarial{T}$, and defined to be an exploration strategy $\algAdversarial[B']{T}$ that maximizes the cost, across all possible values of $B'$:
\begin{equation} \label{eq:ADFS-def0}
\cost{\algAdversarial{T}} = \max_{0 \le B' \le 2\varphi(T)} \cost{\algAdversarial[B']{T}}.
\end{equation}

In the following analysis, it will be convenient for us to use arguments that rely on the fact that $B'$, for our purposes, takes only one of the finite values from $[0,2\varphi(T)]$.
This is due to the above comment, namely, for $\tree{T}{r}$, where $r$ is the root of $T$, we have that $B'=B$ and for any other subtree $\tree{T}{v}$, the values of $B'$ interesting for us depend on the prefix of the route that starts at the root of $T$ and reaches $v$.
Thus, a simple inductive argument allows us to conclude that the value of $B'$ depends on all possible DFS $B$-exploration strategies of $T$ (the number of those is finite).
Hence, we denote by $\allB$ the finite set that consists all values $B'=B-x$ such that there exists a route in $T$ of length $x$ that starts at the root of $T$ and ends at $x$.
Thus, we can restate~\eqref{eq:ADFS-def0}:
\begin{equation*} 
\cost{\algAdversarial{T}} = \max_{B'\in\allB} \cost{\algAdversarial[B']{T}}.
\end{equation*}
Note that it follows from the definition that
\begin{equation*}
\cost{\algAdversarial{T}} \geq \cost{\pdfs{T}}.
\end{equation*}
Intuitively, if $v$ is any node of the tree $T$, then $\algAdversarial{\tree{T}{v}}$ is the worst case scenario of how a DFS $B$-exploration may perform in $\tree{T}{v}$ in terms of the cost; this worst case is understood as considering the worst possible ending point of the route that (in the entire tree $T$) precedes the considered strategy $\algAdversarial{\tree{T}{v}}$.

\begin{lemma} \label{lem:opt-vs-dfs'}
If $v$ is any node of $T$ and $e_1,e_2,\ldots,e_k$ are all downward edges of $v$, then
\begin{equation} \label{eq:opt1}  \cost{\copt{\tree{T}{v}}} = \sum_{1 \le i \le k} \cost{\copt{\tree{T}{e_i}}}, \end{equation}
\vspace{-0.25cm}
\begin{equation} \label{eq:dfs1}  \cost{\algAdversarial{\tree{T}{v}}} \le \sum_{1 \le i \le k} \cost{\algAdversarial{\tree{T}{e_i}}}.\end{equation}
\end{lemma}
\begin{proof}
Informally, equality in~\eqref{eq:opt1} for an optimal solution follows from the fact that $\copt{\tree{T}{v}}$ has the freedom to pick the length of each route to be an arbitrary number in $\allB$.
Any strategies for $\tree{T}{e_1},\tree{T}{e_2}, \ldots, \tree{T}{e_k}$ can be translated into strategy for $\tree{T}{v}$: the latter one is constructed by simply concatenating the former exploration strategies.
Similarly, if one takes an exploration strategy $\copt{\tree{T}{v}}$, then one can assume without affecting its cost that each of its routes has only two occurrences of the root: it is the first and last vertex of the route.
But then, such a strategy $\copt{\tree{T}{v}}$ can be partitioned into the corresponding strategies for the trees $\tree{T}{e_1},\tree{T}{e_2}, \ldots, \tree{T}{e_k}$.

We now prove \eqref{eq:dfs1}.
Consider an exploration strategy $\algAdversarial{\tree{T}{v}}$.
Each route of this strategy is of length at most $2\varphi(T)$.
Obtain an exploration strategy $\cS$ by partitioning each route in $\algAdversarial{\tree{T}{v}}$ in such a way that the concatenation of all routes in $\cS$ equals the concatenation of all routes in $\algAdversarial{\tree{T}{v}}$ and no route in $\cS$ has $v$ as an internal vertex.
(Thus, each route of $\cS$ starts and ends at $v$.)
Note that $\cost{\cS}=\cost{\algAdversarial{\tree{T}{v}}}$.
Now, $\cS$ can be partitioned into $\cS_1,\ldots,\cS_k$ such that $\cS_i$ is a $b_i'$-adversarial DFS exploration strategy of $\tree{T}{e_i}$ for some $b_i'\leq 2\varphi(\tree{T}{v})$, i.e., $\cS_i=\algAdversarial[b_i']{\tree{T}{e_i}}$, for each $i\in\{1,\ldots,k\}$ and the concatenation of $\cS_1,\ldots,\cS_k$ gives $\cS$.
Thus,
\[\cost{\algAdversarial{\tree{T}{v}}} = \cost{\cS} = \sum_{i=1}^k \cost{\algAdversarial[b_i']{\tree{T}{e_i}}}.\]
To conclude the proof, observe that by the definition of adversarial exploration
\[\cost{\algAdversarial[b_i']{\tree{T}{e_i}}} \leq \cost{\algAdversarial{\tree{T}{e_i}}},\quad i\in\{1,\ldots,k\}.\qedhere\]
\end{proof}

The proof of Theorem~\ref{thm:cost} will follow from the following two results (and the fact that $\cost{\pdfs{T}}\leq\cost{\algAdversarial{T}}$.)

\begin{lemma} \label{lem:light}
For any tree $T$, if $T$ is light, then $\cost{\algAdversarial{T}} < 2 \cdot \cost{\copt{T}}.$
\end{lemma}
\begin{proof}
If $T$ is light, then observe that $\algAdversarial{T}$ either consists of one route, in which case $\cost{\algAdversarial{T}}=\cost{\copt{T}}$, or it contains at least two routes but then the second route, having length up to $2\varphi(T)$, will explore all remaining vertices of $T$ since $\varphi(T)\geq\omega(T)$ holds for a light tree.
Thus, in the latter case $\algAdversarial{T}$ has exactly two routes, one of them being strictly shorter than $2\omega(T)$, which gives $\cost{\algAdversarial{T}} < 4\omega(T) \leq 2 \cdot \cost{\copt{T}}$.
\end{proof}

\begin{theorem} \label{thm:induction}
If $T$ is heavy and $r$ is its root, then:
\begin{enumerate} [label={\normalfont{(\roman*)}}, leftmargin=*]
\item\label{it:rec1} if $\hdeg(r) = 1$, then $\cost{\algAdversarial{T}} < 10 \cdot \cost{\copt{T}} -  8 \cdot \varphi(T)$,
\item\label{it:rec2} if $\hdeg(r) \not=1$, then $\cost{\algAdversarial{T}} < 10 \cdot \cost{\copt{T}} - 16 \cdot \varphi(T)$.
\end{enumerate}
\end{theorem}

In order to prove the above Theorem we will first define a special class of heavy trees called \emph{Skinny Tree} which has the following property.
\begin{enumerate}[label={\textbf{(ST)}}]
 \item\label{it:property} 
\textbf{Skinny Tree Property:} \emph{If the root $r$ of $T$ has heavy degree equal to one, then consider the longest path in $T$ that connects $r$ to such a $r'$ that each internal vertex of the path has heavy degree equal to $1$ (i.e. each edge of the path is heavy).
We then require that each vertex of this path, except for $r'$, has at most one light edge incident to it.}
\end{enumerate}

We can show (c.f. Section~\ref{sec:rearrangement-appendix} in the appendix) how to rearrange any tree to have the above property and we also show that: 
\begin{lemma} \label{lem:perturbe}
For any tree T, there exists $\varepsilon>0$ and corresponding tree $\Tprime{\varepsilon}$ such that 
(i) $\Tprime{\varepsilon}$ satisfies Property~\ref{it:property}, and 
(ii) if Theorem~\ref{thm:induction} holds for $\Tprime{\varepsilon}$ then Theorem~\ref{thm:induction} holds for $T$.
\end{lemma}

Due to above result, we can now focus on proving Theorem~\ref{thm:induction} for any tree $T$ with the above-mentioned property in the rest of the paper (Section~\ref{sec:induction}).

\newcommand{\lle}{l}
\subsection{Proof of Theorem~\ref{thm:induction} for Skinny Trees}
\label{sec:induction}
We will proceed by induction on the number of heavy edges in a tree.
This is a valid approach since the parent of a heavy node is also heavy.

For the base case consider $T$ with no heavy edges.
In particular we have that $\hdeg(r)=0$. Denote downward edges at $r$ by $e_1,e_2,\ldots,e_{\lle}$.
For each $i \in\{1,\ldots,\lle\}$, $\tree{T}{e_i}$ is light and hence by Lemma~\ref{lem:light}, $\cost{\algAdversarial{\tree{T}{e_i}}} < 2 \cdot \cost{\copt{\tree{T}{e_i}}}$.
Thus in particular, by \eqref{eq:opt1} and \eqref{eq:dfs1} and $\cost{\copt{T}} \ge 2 \omega(T)$,
\[\cost{\algAdversarial{T}} <  2 \cost{\copt{T}} \le 10 \cost{\copt{T}} - 16\omega(T) \le 10 \cost{\copt{T}} - 16 \varphi(T).\]

For the induction step, we assume that Theorem~\ref{thm:induction} holds for all heavy proper subtrees of $T$.
In what follows we consider two cases: when $\hdeg(r)>1$ and $\hdeg(r)=1$.

\subsubsection*{Case of $\hdeg(r)>1$} 
Let $e_1,\ldots,e_h$ be the heavy downward edges at $r$ and $e'_1, \ldots, e'_{\lle}$ be the light downward edges at $r$.
By induction hypothesis (and precisely Theorem~\ref{thm:induction}\ref{it:rec1}) we obtain
\[\cost{\algAdversarial{\tree{T}{e_i}}} < 10 \cdot \cost{\copt{\tree{T}{e_i}}} - 8 \varphi(\tree{T}{e_i})\]
for each $i\in\{1,\ldots,h\}$.
Then, by \eqref{eq:dfs1} of Lemma~\ref{lem:opt-vs-dfs'}, by Lemma~\ref{lem:light}, and by $\varphi(T) = \varphi(\tree{T}{e_i}))$ (used in this order):
\begin{eqnarray*}
\cost{\algAdversarial{T}} & \leq  & \sum_{i=1}^{h} \algAdversarial{\tree{T}{e_i}} + \sum_{i=1}^{\lle} \algAdversarial{\tree{T}{e_i'}} \\
                           & <  & \sum_{i=1}^{h} \big(10 \cost{\copt{\tree{T}{e_i}}} - 8 \varphi(T) \big) + \sum_{i=1}^{\lle} 2 \cost{\copt{\tree{T}{e'_i}}} \\
                           & \le & 10 \left(  \sum_{i=1}^{h} \cost{\copt{\tree{T}{e_i}}} +  \sum_{i=1}^{\lle} \cost{\copt{\tree{T}{e'_i}}} \right) - 8 \varphi(T) \cdot h \\
                           & \le & 10 \cost{\copt{T}} - 16 \varphi(T).
\end{eqnarray*}
The last inequality is due to \eqref{eq:opt1} and $h=\hdeg(r)\ge 2$.

\subsubsection*{Case of $\hdeg(r)=1$} 
Let $r'$ be the closest descendant of $r$ in $T$ that is heavy and satisfies $\hdeg(r') \neq 1$.
Note that such a vertex always exists and $r'$ is unique.
Let $P$ denote the path connecting $r$ to $r'$.
Additionally, we denote by $e_1,e_2, \ldots, e_{\lle}$ all light edges incident to vertices in $V(P)\setminus\{r'\}$ in the \emph{non-decreasing} order of their potentials.
(We remark here that the subtree rooted at $r'$ has been covered by the base case of the induction and by the case when the heavy degree is greater than one.)
Denote $\varphi_i = \varphi(e_i)$, $i\in\{1,\ldots,\lle\}$.
Due to Lemma~\ref{lem:perturbe} (and more precisely by the fact that thanks to Lemma~\ref{lem:perturbe} we assume that in the tree $T$ all edges $e_1,\ldots,e_{\lle}$ have pairwise different potentials) we have:
\begin{equation*} 
\varphi(r) > \varphi_{\lle} > \cdots > \varphi_2 > \varphi_1 > \varphi(r')=\varphi_0.
\end{equation*}
See Figure~\ref{fig:hs}(a) that illustrates the path $P$ and placements of the edges $e_i$ and the corresponding potentials.
\begin{figure}[htb]
\begin{center}
\includegraphics[scale=0.8]{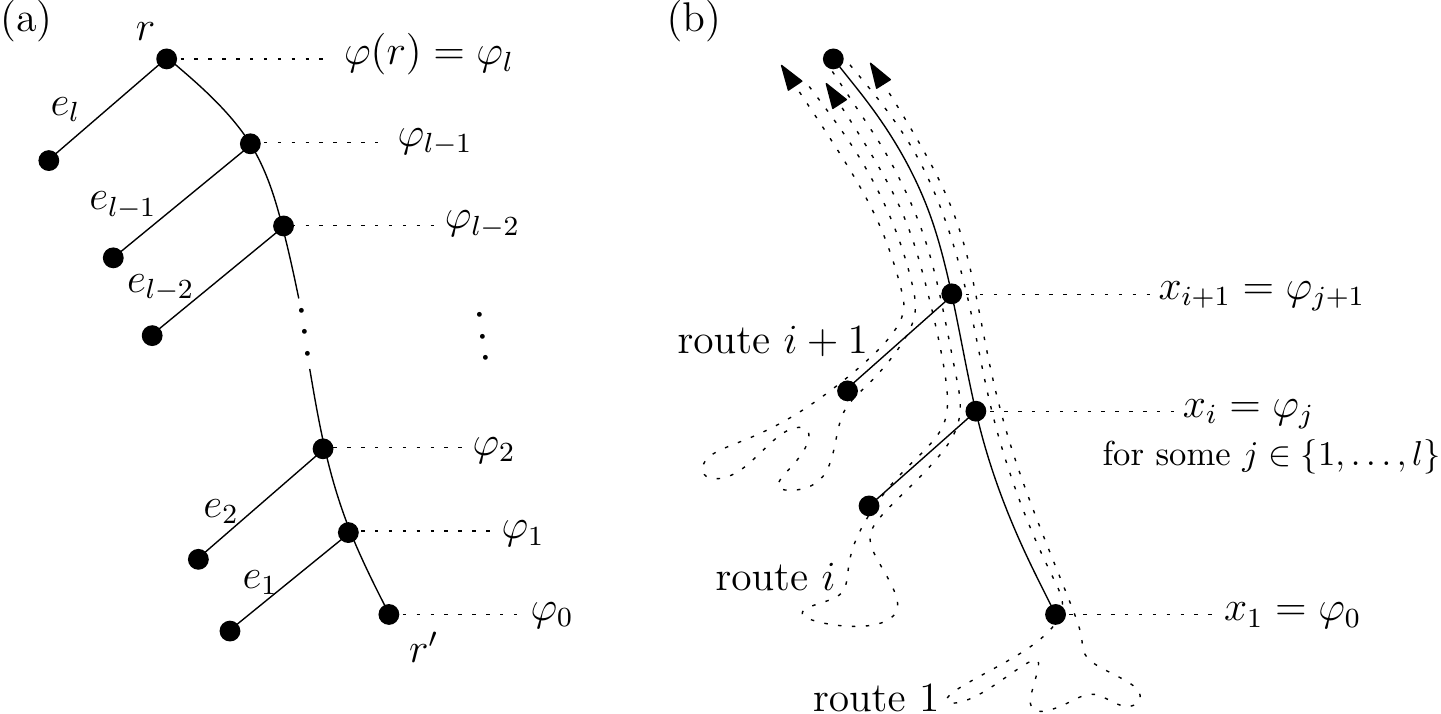}
\caption{(a) the path $P$ (although it may contain vertices of degree two, they are trivial from the point of view of the analysis, thus omitted in the picture);
         (b) illustration of $x_i$'s. Note that the first route (and possibly arbitrarily many more routes) always reaches $r'$}
\label{fig:hs}
\end{center}
\end{figure}

We take for brevity
\[w_i = \omega(\tree{T}{e_i}) = \frac12\cost{\copt{\tree{T}{e_i}}} \text{\quad for\ } i \in \{1,\ldots,\lle\},\quad w_{0} = \frac12\cost{\copt{\tree{T}{r'}}}.\]

Let $c$ be the number of routes in $\copt{T}$, and we denote by $x_i$ the lowest potential $i$-th route reached on the path $P$ (we ignore potentials it reached in subtrees --- see Figure~\ref{fig:hs}(b)), where the routes are without loss of generality ordered so that
$x_1 \le x_2 \le \cdots \le x_c.$

Consider $j\in\{1,\ldots,c\}$.
Informally speaking, the first $j$ routes of $\copt{T}$ need to cover all subtrees $\tree{T}{e_i}$ such that the potential of the higher endpoint of $e_i$ is \emph{strictly} smaller than $x_j$; otherwise some vertices would not be visited by $\copt{T}$.
The total weight of these subtrees is $\sum_{i : \varphi_i < x_j} w_i$.
Observe, that the total length of all parts of an $i$-th route that do not belong to the path $P$ is at most $2x_i$.
Thus, the above total weight of the above-mentioned subtrees satisfies
\begin{equation} \label{eq3}
\sum_{i : \varphi_i < x_j} w_i \le (x_1 + \cdots + x_{j-1}) \quad\text{for each } j\in\{1,\ldots,c+1\},
\end{equation}
where we denote $x_{c+1} = +\infty$ for the sake of simplicity.

We are now interested in bounding the cost of $\algAdversarial{T}$ on the path $P$ with respect to
$\sum_{i=1}^c 2(\varphi(r)-x_i),$
that is, with respect to the cost of $\copt{T}$ on the path $P$.
To do this, we start by comparing $x_1,\ldots,x_c$ with $y_1,\ldots,y_d$ chosen by an appropriate greedy procedure:
\begin{equation}\label{eq4}
y_j = \min\big\{  y : \sum_{i : \varphi_i \le y} w_i > y_1 + \ldots + y_{j-1}\big\},
\end{equation}
\[d = \min\left\{ i : y_1 + \cdots + y_i \ge w_0+w_1+\ldots+w_m \right\},\]
in other words, assigning $y_j$ to be the first value where \eqref{eq3} is violated given only $y_1,\ldots,y_{j-1}$.
(Notice that from the definition we have that always $y_j \in \{\varphi_0,\varphi_1,\ldots,\varphi_{\lle}\}$. Moreover, if $y_j > \varphi_0$, then $y_{j+1} > y_j$.)
We obtain the following lemma which says that, across all sequences satisfying \eqref{eq3}, $y_i$ takes maximal values:
\begin{lemma} \label{lem:y-min}
It holds that $d \le c$ and $y_j \ge x_j$ for each $j\in\{1,\ldots,d\}$.
\end{lemma}
\begin{proof}
We proceed by induction on $j$. 

For the inductive base, we have $x_1 = y_1 = \varphi(r')$.
For the inductive step, assume the claim holds for all indices smaller than $j$.
Suppose for a contradiction that $x_j > y_j$.
Then,
\[\sum_{i : \varphi_i \le y_j} w_i \le \sum_{i : \varphi_i < x_j} w_i\]
and by \eqref{eq3} and the inductive assumption applied for all indices smaller than $j$,
\[\sum_{i : \varphi_i < x_j} w_i\le y_1 + \ldots + y_{j-1}.\]
These two inequalities give a contradiction with $\eqref{eq4}$.
\end{proof}

Observe that $\algAdversarial{T}$ first traverses (in that order) some subset of light subtrees $\tree{T}{e_i}$, whose indices we denote by $H \subseteq \{1,\ldots,\lle\}$, in a decreasing order of their indices.
The above routes, none of which contains $r'$, will form the \emph{first part of} $\algAdversarial{T}$.
Then, all vertices of $\tree{T}{r'}$ are visited (to those routes of $\algAdversarial{T}$ we refer at the \emph{second part of} $\algAdversarial{T}$) and following that, remaining light subtrees $\tree{T}{e_i}$ for $i \in H' = \{1,\ldots,\lle\}\setminus H$, in an increasing order of their indices (\emph{third part of} $\algAdversarial{T}$).
Note that there may exist a route that has a non-empty intersection with a tree $\tree{T}{e_i}$, $i\in H\cup H'$, and also contains $r'$ --- this route belongs by definition to the second part of $\algAdversarial{T}$.

Denote by $z_1,\ldots,z_p$  the lowest potentials reached by subsequent routes in $\algAdversarial{T}$ on the path $P$ (ignoring as before the potentials they reach in subtrees), only in the first part of $\algAdversarial{T}$ in the reversed order of entering $T$:
\begin{equation} \label{eq:zi-def}
\varphi(r) \ge z_p \ge \cdots \ge z_1 > \varphi(r').
\end{equation}
We note that each route in the first part of $\algAdversarial{T}$ visits subtrees with a continuous segment of indices from $H$, that is $\tree{T}{e_i}$'s for $i \in H \cap\{i_1,\ldots, i_2\}$ for some integers $i_1,i_2$.
We, due to the weight of a light tree, its vertices belong to at most two different routes.

\begin{lemma}
\label{lem:2step}
If for some $i,j$ there is $z_i > y_j$, then $z_{i+1} \ge y_{j+1}$.
\end{lemma}
\begin{proof} 
If $y_{j+1}  = y_j $ then the claim follows immediately from the fact that, by~\eqref{eq:zi-def}, $z_{i+1} \geq z_i \geq y_j = y_{j+1}$.
Similarly, if $z_i \ge y_{j+1}$ then 
 by~\eqref{eq:zi-def} we have $z_{i+1} \geq z_i \ge y_{j+1}$.
Thus, assume that $y_{j+1} \ge z_i$ and $y_{j+1} > y_j$. 

Let $a>b$ be indices such that $y_j = \varphi_a$ and $y_{j+1}= \varphi_b$.
We have from the way $y_j$ is selected in \eqref{eq4}:
\[w_0 + \cdots + w_b > y_1 + \cdots + y_j \ge w_0 + \cdots + w_{b-1},\]
\[w_0 + \cdots + w_a > y_1 + \cdots + y_{j-1} \ge w_0 + \cdots + w_{a-1}.\]
Thus,
\[y_j > w_{a+1} + \cdots + w_{b-1}.\]
This inequality, informally speaking, certifies that the total weight of all subtrees $\tree{T}{e_s}$ with $s\in\{a+1,\ldots,b-1\}$ is smaller than $y_j$.
By assumption $z_i>y_j$.
By the definition of the sequence $z_1,\ldots,z_p$, the length of the $i$-th route in $\algAdversarial{T}$ restricted to the subtrees $\tree{T}{e_s}$ is at least $2z_i$.
(Note that we are not using the fact that this route may avoid some subtrees $\tree{T}{e_s}$ with $s\in\{a+1,\ldots,b-1\}$ as we analyze the first part of $\algAdversarial{T}$ which `avoids' each subtree $\tree{T}{e_s}$ with $s\in\{a+1,\ldots,b-1\}\cap H'$.)
Thus, the $i$-th route of $\algAdversarial{T}$ visits the node of $P$ at potential $\varphi_b$ and hence $z_{i+1} \ge \varphi_b = y_{j+1}$ as required in the lemma.
\end{proof}

We are now ready to bound the total cost of $\algAdversarial{T}$ with relation to $\copt{T}$. 
The cost of $\copt{T}$ can be decomposed:
\begin{equation} \label{eq:opt-decomposed}
\cost{\copt{T}}=O_{light} + O_{deep} + O_{path} + O_{flat},
\end{equation}
where:
\begin{description}
 \item $O_{light}$ --- is the cost restricted to light subtrees $\tree{T}{e_i}$ with $i\in\{1,\ldots,\lle\}$,
 \item $O_{deep}$  --- is the cost restricted to the subtree $\tree{T}{r'}$,
 \item $O_{path}$  --- is the cost restricted to the path $P$ and the routes that do not contain $r'$, and
 \item $O_{flat}$  --- is the cost restricted to the path $P$ and routes that do contain $r'$.
\end{description}
Similarly, we express the cost of $\algAdversarial{T}$ as a sum:
\begin{equation} \label{eq:dfs-decomposed}
\cost{\algAdversarial{T}} = D_{light} + D_{deep} + D_{desc} + D_{flat} + D_{asc},
\end{equation}
where
\begin{description}
 \item $D_{light}$ --- is the cost of $\algAdversarial{T}$ restricted to light subtrees $\tree{T}{e_i}$ with $i\in\{1,\ldots,\lle\}$,
 \item $D_{deep}$ --- is the cost restricted to the subtree $\tree{T}{r'}$,
 \item $D_{desc}$ --- is the cost restricted to the path $P$ in the first part of $\algAdversarial{T}$,
 \item $D_{flat}$ --- is the cost restricted to $P$ and the routes that contain $r'$ (i.e., the second part of $\algAdversarial{T}$), and
 \item $D_{asc}$ --- is the cost restricted to the path $P$  in the third part of $\algAdversarial{T}$.
\end{description}

By Lemma~\ref{lem:light}, $\cost{\algAdversarial{\tree{T}{e_i}}} < 2 \cost{\copt{\tree{T}{e_i}}}$ for each $i\in\{1,\ldots,\lle\}$ and therefore
\begin{equation} \label{eq:Dlight}
D_{light} < 2\cdot O_{light}.
\end{equation}

Denote by $s$ the smallest index such that $y_{s+1} > \varphi_0 = y_s$ and let $H=\{j_1,\ldots,j_q\}$, $j_1<j_2<\cdots<j_q$.
Recall that $H$ is the set of indices $i$ such that $\tree{T}{e_i}$ is covered in the first part of $\algAdversarial{T}$.
Since $z_{j_1}>y_s$, by iteratively applying Lemma~\ref{lem:2step} we obtain that
\[z_{j_i} \geq y_{s+i-1} \textup{ for each }i\in\{1,\ldots,q\}.\]
Therefore, by Lemma~\ref{lem:y-min},
\begin{eqnarray} \label{eq:Ddesc}
\begin{aligned}
D_{desc} & =  2\sum_{i=1}^q(\varphi(T)-z_{j_i}) \leq 2\sum_{i=1}^q (\varphi(T)-y_{s+i-1}) \\
         & \leq  2(\varphi(T) - \varphi_0) + 2\sum_{i=2}^q (\varphi(T)-x_{s+i-1}) \leq 2\omega(P) + O_{path}.
\end{aligned}
\end{eqnarray}
By an analogous analysis, the same bound holds for the third part of $\algAdversarial{T}$: 
\begin{equation} \label{eq:Dasc}
D_{asc} \le 2\omega(P) + O_{path}.
\end{equation}
By the inductive assumption we have
\begin{equation} \label{eq:Ddeep}
D_{deep} < 10 \cdot O_{deep} - 16 \varphi(r').
\end{equation}
We also get the following bounds by analyzing how much each particular route can overlap with $\tree{T}{r'}$.
The first one follows from an observation that each route having a non-empty intersection with $\tree{T}{r'}$ may have length restricted to $\tree{T}{r'}$ at most $2\varphi(r')$.
Thus, there exist at least $O_{deep}/(2\varphi(r'))$ such routes in $\algAdversarial{T}$ intersecting $\tree{T}{r'}$ and each such route contributes at least $2\omega(P)$ to $O_{flat}$.
Hence,
\begin{equation} \label{eq:Oflat}
O_{flat} \ge \frac{\omega(P)}{\varphi(r')} \cdot O_{deep}.
\end{equation}
As for an upper bound, there exist at most $\lceil D_{deep}/(2\varphi(r'))\rceil+1$ routes in $\pdfs{T}$ that contain $r'$ (note that the first such route may include no other vertices except for $r'$ from $\tree{T}{r'}$).
Thus,
\begin{equation} \label{eq:Dflat}
D_{flat} \le 2\omega(P) \cdot \lceil D_{deep}/(2\varphi(r')) + 1 \rceil \le 4 \omega(P) + \omega(P) \cdot D_{deep}/\varphi(r').
\end{equation}
Equations \eqref{eq:Dflat}, \eqref{eq:Ddeep} and \eqref{eq:Oflat}, used in that order, give us
\begin{eqnarray*}
D_{deep} + D_{flat} & \leq &  4\omega(P) +  \left(\frac{\omega(P)}{\varphi(r')}+1\right)D_{deep} \\
                    & \leq & 4\omega(P) + \left(\frac{\omega(P)}{\varphi(r')}+1\right)(10\cdot O_{deep} - 16\varphi(r')) \\
                    & \leq & 4\omega(P) - 16\left(\omega(P)+\varphi(r')\right) + 10\cdot O_{deep} \left(\frac{\omega(P)}{\varphi(r')}+1\right) \\
                    & \leq & 10 O_{deep} + 10 O_{flat} - 12 \omega(P) - 16\varphi(r').
\end{eqnarray*}
Then, by~\eqref{eq:Dlight}, \eqref{eq:Ddesc} and \eqref{eq:Dasc} we have
\[D_{light} + D_{asc} + D_{desc} \le 2O_{light} + 4\omega(P) + 2 O_{path}.\]
The last two inequalities,~\eqref{eq:opt-decomposed},~\eqref{eq:dfs-decomposed} and $\varphi(r)=\omega(P)+\varphi(r')$ finally give
\begin{eqnarray}
\cost{\algAdversarial{T}} & \le & 10 O_{deep} + 10 O_{flat} + 2O_{light} + 2O_{path} - 8\omega(P) - 16\varphi(r') \nonumber \\ 
         & \le & 10 \cost{\copt{T}} - 8 \varphi(r), \nonumber
\end{eqnarray}
which completes the inductive proof of Theorem~\ref{thm:induction}.

\section{Conclusions and open problems}

Our strategy $\pdfs{T}$ achieves the asymptotically minimum number of routes and also minimizes the cost up to a small constant. In particular, we provided an upper bound of $10$ for the competitiveness of any online piecemeal exploration strategy. A trivial lower bound for the same problem is $3/2$ (Consider the tree with three branches of lengths $B/2$, $B/2$ and $B$, respectively, starting from the root: any online algorithm may cover the tree with 3 routes, while the optimal is 2 routes). This leaves a gap between the lower and upper bounds and the interesting open question is whether the strategy $\pdfs{T}$ is the best possible algorithm? Another open problem is to analyze similar strategies in other, more general, classes of graphs instead of trees.

\bibliography{batteries}

\appendix
\begin{center}
\textbf{\Large Appendix}
\end{center}

\section{Tree rearrangement} 
\label{sec:rearrangement-appendix}

This section is devoted to proving Lemma~\ref{lem:perturbe}. 
We start with an informal description providing a high level intuition that gives an overview of this section.
Our first step (Section~\ref{sec:Tprime-construction}) is to construct a tree $\Tprime{\varepsilon}$, $\varepsilon\in\mathbb{R}^+$, based on $T$  such that $\Tprime{\varepsilon}$ satisfies property~\ref{it:property}.
We will also need that $\Tprime{\varepsilon}$ `resembles' $T$ in the following way: a search strategy is valid for $T$ if and only if a `very similar' strategy is valid for $\Tprime{\varepsilon}$.
To simplify the statements considerably, it will be convenient to encode strategies in an uniform way so that we can apply the same strategy for both trees, without going into the details of tedious but straightforward conversions between strategy for $T$ and strategy for $\Tprime{\varepsilon}$.
We thus define (Section~\ref{sec:Tprime-epsilon}) a collection $\cC$ of all possible strategies (including adversarial ones and those that are not feasible for either $T$ or $\Tprime{\varepsilon}$ because they contain routes that are too long or do not visit all vertices).
Then in Section~\ref{sec:Tprime-analysis} we select the right value of $\varepsilon$.
The value of $\varepsilon$ and the construction of $\Tprime{\varepsilon}$ will ensure that a strategy in $\cC$ is valid for $T$ if and only if it is valid for $\Tprime{\varepsilon}$.
We then finally provide the main result of this section (Lemma~\ref{lem:perturbe})  states that, again thanks to the choice of $\varepsilon$, if Theorem~\ref{thm:induction} holds for $\Tprime{\varepsilon}$, then it holds for $T$, thus allowing us to restrict only to trees satisfying property~\ref{it:property}.

\subsection{The construction of $\Tprime{\varepsilon}$} \label{sec:Tprime-construction}
We now construct the tree $\Tprime{\varepsilon}=(V(T)\cup X,E(\Tprime{\varepsilon}),\omega')$ based on $T=(V(T),E(T),\omega)$ and the construction depends on a parameter $\varepsilon>0$ that will be fixed later.
We now impose only a condition on $\varepsilon$ that is needed for the construction itself to be valid:
\[\varepsilon < \min\{\omega(e)\myst e\in E(T)\}.\]
Select an arbitrary vertex $u$ in $T$ with $\hdeg(u)=1$ (denote by $\{u,v\}$ the heavy downward edge at $u$) and $d>1$ light downward edges $e_1,\ldots,e_d$ at $v$.
Subdivide the edge $e$ (see Figure~\ref{fig:perturbe} for an illustration) by replacing it by a path $P=(u=v_0,v_1,\ldots,v_{d-1},v_d=v)$ with the following edge lengths: the first $d-1$ edges have length $\varepsilon/(d-1)$, i.e., $\omega'(v_i,v_{i+1})=\varepsilon/(d-1)$ for each $i\in\{0,\ldots,d-2\}$, and for the last edge we set $\omega'(\{v_{d-1},v\})=\omega(\{u,v\})-\varepsilon$.
(Note that this preserves the distance between $u$ and $v$.)
Then, the weight of each edge $e_i$ decreases by $\varepsilon$, $\omega'(e_i)=\omega(e_i)-\varepsilon$ and the higher endpoint of $e_i$ in $\Tprime{\varepsilon}$ becomes $v_{i-1}$, $i\in\{1,\ldots,d\}$.
(Note that this ensures that the distance between $u$ and the lower endpoint of $e_i$ or the distance between two children of $u$ is not greater in $\Tprime{\varepsilon}$ than in $T$.)
This construction allows us to assume (by permuting the edges $e_1,\ldots,e_d$ appropriately) that the DFS traversal of $\Tprime{\varepsilon}$ visits the edges $e$ and $e_1,\ldots,e_d$ in the same order both in $T$ and in $\Tprime{\varepsilon}$, ensuring that Condition~\ref{it:P2} is satisfied.
\begin{figure}[htb]
\begin{center}
\includegraphics[scale=0.6]{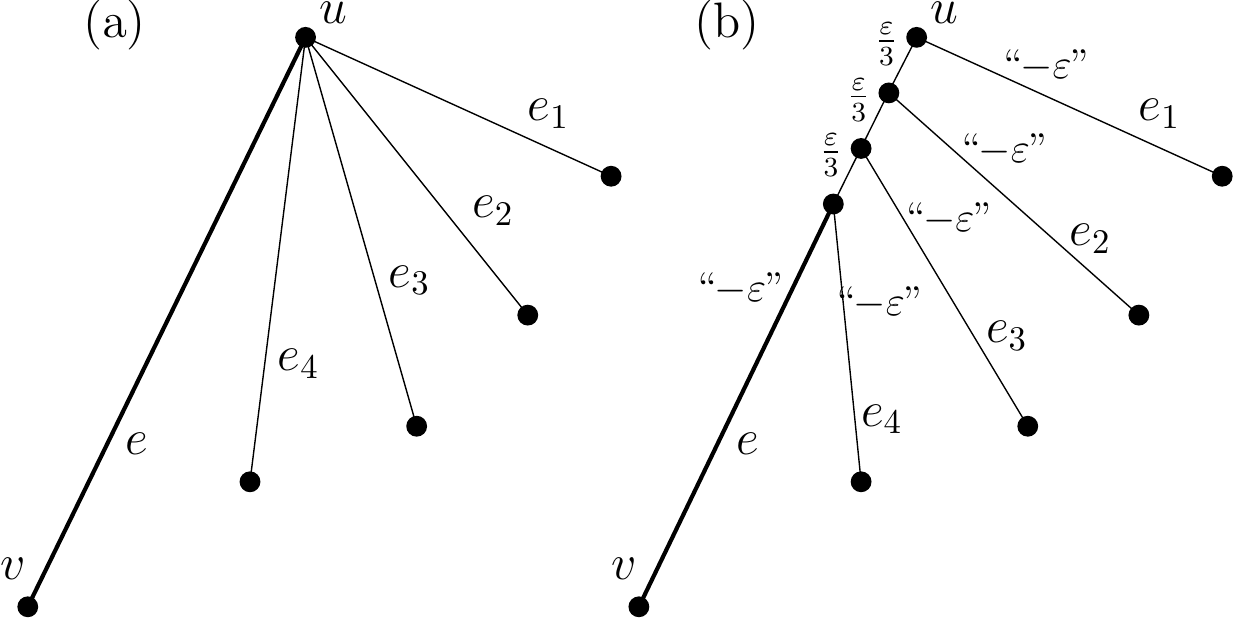}
\caption{(a) the node $u$ of $T$ with $\hdeg(u)=1$ and light downward edges $e_1,\ldots,e_d$, $d=4$;
         (b) the corresponding edges in $\Tprime{\varepsilon}$ (the ``$-\varepsilon$'' indicated that the weight of the edge is $\varepsilon$ smaller than the weight of the corresponding edge in $T$)}
\label{fig:perturbe}
\end{center}
\end{figure}
Since the vertex $v$ is selected arbitrarily, we repeat the above modification for each such vertex $v$ obtaining the final tree $\Tprime{\varepsilon}$.

\medskip
We will require the following conditions to be satisfied:
\begin{enumerate} [label={\normalfont{(P\arabic*)}},leftmargin=30pt]
 \item\label{it:P1} there are at most two downward edges at each vertex in $\Tprime{\varepsilon}$ with $\hdeg(v)=1$,
 \item\label{it:P2} there exists a DFS traversal of $\Tprime{\varepsilon}$ that visits the vertices in $V(T)$ in the same order as the DFS traversal that we have fixed for $T$ in this work,
 \item\label{it:P3} $\cS\in\cC$ is feasible for $T$ if and only if $\cS$ is feasible for $\Tprime{\varepsilon}$.
\end{enumerate}
Property~\ref{it:P3} will be proved in Lemma~\ref{lem:classEps} and we now note:
\begin{observation} \label{obs:P1P2}
For each $\varepsilon>0$, the tree $\Tprime{\varepsilon}$ satisfies Conditions~\ref{it:P1} and~\ref{it:P2}.
\qed
\end{observation}
\begin{observation} \label{obs:P2ADFS}
The strategies $\algAdversarial{T}$ and $\algAdversarial{\Tprime{\varepsilon}}$ visit the nodes in $V(T)$ in the same order.
\qed
\end{observation}

\subsection{Finding the right value of $\varepsilon$} \label{sec:Tprime-epsilon}

We define a \emph{potential route} as a following pair: $R'=(L,v)$, where $L=(l_1,\ldots,l_p)$ is a sequence of leaves and $v\in V(\Tprime{\varepsilon})$.
Then, $R'$ \emph{translates} to a route $R$ in $T$ as a concatenation of the following paths (in this order): the path from $r$ to $l_1$, the path from $l_i$ to $l_{i+1}$, $i=1,\ldots,p-1$, the path from $l_p$ to the closest ancestor $x$ of $v$ that belongs to $V(T)$ and finally the path from $x$ to $r$.
The length of $R$ is
\[\routelen{R}=d(r,l_1)+d(l_p,x)+d(x,r)+\sum_{1\leq i<p}d(l_i,l_{i+1}).\]
$R'$ translates to a route in $\Tprime{\varepsilon}$ in the same way, except that take $x=v$, i.e., $v$ is not replaced by the ancestor.
Then, a \emph{potential strategy} is a sequence consisting of at most $j$ potential routes, $j\in\{1,\ldots,p\}$, where $p$ is the number of leaves in $T$.

Note that a potential strategy may not translate to a valid $B$-exploration strategy for $T$ or $\Tprime{\varepsilon}$ because some nodes may not be explored and some routes may be too long.
A potential strategy is \emph{feasible} for $T$ (respectively $\Tprime{\varepsilon}$) if it translates to a valid $B$-exploration strategy for $T$ (respectively $\Tprime{\varepsilon}$).
We denote by $\cC$ a collection of all potential strategies.
Clearly, the size of $\cC$ is finite.

\medskip
We conclude with the following:
\begin{observation} \label{obs:potential-strategies}
For any route $R$ that may appear in $\algAdversarial{T}$, $\copt{T}$, $\algAdversarial{\Tprime{\varepsilon}}$ and $\copt{\Tprime{\varepsilon}}$ there exists a potential route that translates to $R$.
$\qed$
\end{observation}

\subsection{The analysis of $\Tprime{\varepsilon}$} \label{sec:Tprime-analysis}

In this section we argue that the construction of $\Tprime{\varepsilon}$ `preserves' the problem: the minimum costs of adversarial DFS explorations of both $T$ and $\Tprime{\varepsilon}$, as well as $\copt{T}$ and $\copt{\Tprime{\varepsilon}}$ remain close to each other for $\varepsilon$ small enough.
Intuitively speaking, this follows from a `continuity argument' formalized in the remaining part of this section.

\medskip
For the tree $\Tprime{\varepsilon}$ we define an interval denoted by $\classEps=(0,y)$, $y<\min\{\omega(e)\myst e\in E(T)\}$, such that for each $\varepsilon\in (0,y)$, Condition~\ref{it:P3} holds for $\Tprime{\varepsilon}$.
We now prove that this interval is well defined, that is, $y>0$.
(Note that for $\varepsilon=0$, $\Tprime{\varepsilon}$ and $T$ are the same.)

\begin{lemma} \label{lem:classEps}
It holds $\classEps\neq\emptyset$.
\end{lemma}
\begin{proof}
First we argue that there exists $y>0$ such that $\Tprime{y}$ fulfills Condition~\ref{it:P3}.
We select $y$ based on the tree $T$ and the collection $\cC$.
Consider any $B'\in\allB$.
The number of potential strategies in $\cC$ is finite, and hence the number of potential strategies in $\cC$ that do not translate to feasible ones for $T$ (denote subset of those by $\cC'$) is also finite.
For each $\cS\in\cC'$, define its \emph{deficiency} $x(\cS)$ as follows: the $x(\cS)$ is the maximum value such that either the length of the first route in $\cS$ is $B'+x(\cS)$ or the length of some other route in $\cS$ in $T$ is $B+x(\cS)$.
Intuitively, $\cS$ does not translate to a feasible $B'$-adversarial $B$-exploration strategy for $T$ because one of its routes exceeds the allowed length by $x(\cS)$ and no route exceeds it by more than $x(\cS)$.
Take
\[y:=\frac{1}{2n(m+1)}\cdot\min\left\{\min\{\omega(e)\myst e\in E(T)\},\min\{x(\cS')\myst \cS'\in\cC'\}\right\},\]
where $m$ is the number edges in $T$.
Since by definition, $x(\cS)>0$ for each $\cS\in\cC'$, we obtain that $y>0$.
Also, no route in any $\cS\in\cC'$ traverses an edge more than $2n$ times and hence the length of any route of $\cS$ in $\Tprime{y}$ decreases by at most $2ynm$ with respect to its length in $T$.
This implies that some route $R$ of $\cS$ has length in $\Tprime{y}$ at least
\begin{equation} \label{eq:Bx}
\tilde{B}+x(\cS)-2ynm\geq \tilde{B}+2yn(m+1)-2ynm>\tilde{B},
\end{equation}
where take $\tilde{B}=B'$ if $R$ is the first route in $\cS$ and $\tilde{B}=B$ otherwise.
Therefore, $\cS$ remains unfeasible in $\Tprime{y}$.
Since exploration strategy that is feasible in $T$ remains feasible in $\Tprime{y}$ (recall that the length of each route is smaller in $\Tprime{y}$ than in $T$), we have that Condition~\ref{it:P3} holds for $\Tprime{y}$.

Finally, observe that substituting $y$ by any value $\varepsilon$ smaller than $y$ in the left hand side of~\eqref{eq:Bx} keeps this equation true.
Therefore, Condition~\ref{it:P3} is satisfied by $\Tprime{\varepsilon}$ for each $\varepsilon\in(0,y)$, which completes the proof.
\end{proof}

Before we state the main lemma of this section, we prove these technical bounds:
\begin{lemma} \label{lem:pert-requirements}
For each $\varepsilon\in\classEps$ it holds:
\begin{enumerate} [label={\normalfont{(\roman*)}},leftmargin=30pt]
\item\label{it:per1} $\cost{\algAdversarial{T}}\leq\cost{\algAdversarial{\Tprime{\varepsilon}}}+4\varepsilon n^2$,
\item\label{it:per2} $\varphi(\Tprime{\varepsilon})\leq \varphi(T)$, and
\item\label{it:per3} $\cost{\copt{\Tprime{\varepsilon}}}\leq\cost{\copt{T}}$.
\end{enumerate}
\end{lemma}
\begin{proof}
By Observation~\ref{obs:P2ADFS}, both $\algAdversarial{T}$ and $\algAdversarial{\Tprime{\varepsilon}}$ visit the leaves of both trees in the same order.
Consider any edge $\{u,v\}$ in $T$ such that $v$ is its lower endpoint.
We have that there exists a \emph{corresponding} edge $\{u',v\}$ in $\Tprime{\varepsilon}$.
Moreover, $\omega(\{u,v\})\leq\omega'(\{u',v\})+\varepsilon$.
Since for each edge traversal in $\algAdversarial{T}$, a traversal of the corresponding edge occurs in $\algAdversarial{\Tprime{\varepsilon}}$ due to Condition~\ref{it:P3}, we obtain that if $t$ is the number of edge traversals in $\algAdversarial{T}$, then
\[\cost{\algAdversarial{T}} \leq \cost{\algAdversarial{\Tprime{\varepsilon}}} + t\varepsilon \leq \cost{\algAdversarial{\Tprime{\varepsilon}}} + 4n^2\varepsilon.\]
The latter inequality follows from bounding each route in $\algAdversarial{\Tprime{\varepsilon}}$ to have at most $2n$ edges, bounding the number of routes by $n$ and observing that an edge is traversed at most twice in each route.

Property~\ref{it:per2} is a direct consequence of the construction of $\Tprime{\varepsilon}$: any path in $T$ corresponds to a path in $\Tprime{\varepsilon}$ that connects the same common nodes and contains all common edges of the original path.

By construction of $\Tprime{\varepsilon}$, the distance between any common nodes $u\in V(T)$ and $v\in V(T)$ is not greater in $T$ than in $\Tprime{\varepsilon}$, which immediately gives~\ref{it:per3}.
\end{proof}


\noindent
\begin{lemma}
There exists $\varepsilon>0$ such that if Theorem~\ref{thm:induction} holds for $\Tprime{\varepsilon}$, then Theorem~\ref{thm:induction} holds for $T$. 
\end{lemma}
\begin{proof}
We will analyze condition~\ref{it:rec1} in Theorem~\ref{thm:induction} and the proof for~\ref{it:rec2} is analogous as we note at the end of the proof.

Define a parameter $\tau(n')$ to be the maximum number for which an inequality
\[\cost{\algAdversarial{T'}} \leq 10 \cdot \cost{\copt{T'}} -  8 \cdot \varphi(T') - \tau(n')\]
holds for each tree $T'$ on at most $n'$ nodes that satisfies Condition~\ref{it:rec1}.
Since the number of such trees is finite and the number of potential strategies in $\cC$ is finite for each tree $T'$, we obtain that $\tau(n')>0$.
By Lemma~\ref{lem:classEps}, there exists $\varepsilon\in\classEps$ such that $\varepsilon>0$ and
\[\varepsilon<\frac{\tau(2n)}{4n^2},\]
$n=\card{V(T)}$, such that $\Tprime{\varepsilon}$ satisfies Condition~\ref{it:P3}.
Suppose that Theorem~\ref{thm:induction}\ref{it:rec1} holds for $\Tprime{\varepsilon}$.
Then, by Lemma~\ref{lem:pert-requirements} (in particular, \ref{it:per1} of Lemma~\ref{lem:pert-requirements} is used to obtain the first inequality below and \ref{it:per2} and~\ref{it:per3} are used to obtain the third inequality below) we get:
\begin{eqnarray}
\cost{\algAdversarial{T}} & < & \cost{\algAdversarial{\Tprime{\varepsilon}}} + 4\varepsilon n^2 \nonumber \\
               & \le & 10 \cdot \cost{\copt{\Tprime{\varepsilon}}} -  8 \cdot \varphi(\Tprime{\varepsilon}) + 4\varepsilon n^2 - \tau(2n) \nonumber \\
               & \le & 10 \cdot \cost{\copt{T}} -  8 \cdot \varphi(T)  + 4\varepsilon n^2 - \tau(2n) \nonumber \\
               & <   & 10 \cdot \cost{\copt{T}} -  8 \cdot \varphi(T). \nonumber 
\end{eqnarray}
We can conduct the same argument for Theorem~\ref{thm:induction}\ref{it:rec2} with the same value of $\varepsilon$.
Hence we obtain that Theorem~\ref{thm:induction} holds for $T$.
\end{proof}

\end{document}